\documentclass[letter, oneside, 11pt, article]{memoir}

\usepackage{amssymb}
\usepackage{amsmath}
\usepackage{amsfonts}
\usepackage{amsthm}

\usepackage[english]{babel}

\usepackage{graphicx}

\usepackage[plain, english]{fancyref}

\usepackage{bbm}
\usepackage{mathrsfs}
\usepackage{stmaryrd}

\usepackage{ifthen}

\usepackage{xspace}

\usepackage{calc}
\usepackage{color}

\newcommand{\mathset}[1]{\mathbbm{#1}}

\newcommand{\setN}{\mathset{N}}

\newcommand{\setR}{\mathset{R}}
\newcommand{\setC}{\mathset{C}}

\newcommand{\abs}[2][]{#1\lvert #2#1\rvert}

\renewcommand{\d}{\partial}
\newcommand{\dd}[2]{\frac{\partial #1}{\partial #2}}

\renewcommand{\phi}{\varphi}
\renewcommand{\epsilon}{\varepsilon}
\renewcommand{\otimes}{\varotimes}

\newcommand{\sotimes}{\mathbin{\raise1.5pt\hbox{
      $\scriptscriptstyle\otimes$}}}

\DeclareMathOperator{\Aut}{Aut}

\DeclareMathOperator{\Ext}{Ext}
\DeclareMathOperator{\Int}{Int}

\theoremstyle{plain}

\newtheorem{proposition}{Proposition}
\newtheorem{theorem}{Theorem}

\newtheorem{definition}{Definition}

\theoremstyle{definition}

\newtheorem{example}{Example}

\makeatletter       

\let\@xp\expandafter 
\newcommand\DefineFancyrefPrefix[2]{%
  \@namedef{fancyref#1labelprefix}{#1}%
  \@namedef{Fref#1name}{#2}%
  \@namedef{fref#1name}{\MakeLowerCase{\@nameuse{Fref#1name}}}%
  \def\@style{vario}%
  \@xp\@xp\@xp\frefformat\@xp\@xp\@xp\@style\@xp\csname
  fancyref#1labelprefix\endcsname
  {%
    \@nameuse{fref#1name}\fancyrefdefaultspacing##1##3%
  }%
  \@xp\@xp\@xp\Frefformat\@xp\@xp\@xp\@style\@xp\csname
  fancyref#1labelprefix\endcsname
  {%
    \@nameuse{Fref#1name}\fancyrefdefaultspacing##1##3%
  }%
  \def\@style{plain}%
  \@xp\@xp\@xp\frefformat\@xp\@xp\@xp\@style\@xp\csname
  fancyref#1labelprefix\endcsname
  {%
    \@nameuse{fref#1name}\fancyrefdefaultspacing##1%
  }%
  \@xp\@xp\@xp\Frefformat\@xp\@xp\@xp\@style\@xp\csname
  fancyref#1labelprefix\endcsname
  {%
    \@nameuse{Fref#1name}\fancyrefdefaultspacing##1%
  }%
}
\makeatother

\DefineFancyrefPrefix{ax}{Axiom}
\DefineFancyrefPrefix{cor}{Corollary}
\DefineFancyrefPrefix{lem}{Lemma}
\DefineFancyrefPrefix{prop}{Proposition}
\DefineFancyrefPrefix{thm}{Theorem}
\DefineFancyrefPrefix{def}{Definition}
\DefineFancyrefPrefix{rem}{Remark}
\DefineFancyrefPrefix{ex}{Example}

\makepagestyle{niels}

\makeoddhead{niels}{}{}{\thepage}
\setlrmarginsandblock{3.9cm}{3.9cm}{*} 
\checkandfixthelayout[nearest]

\chapterstyle{article}
\begin{document}

\title{A Universal Formula for Deformation Quantization \\ on K\"ahler
  Manifolds}

\author{Niels Leth Gammelgaard}

\maketitle

\begin{abstract}
  We give an explicit local formula for any formal deformation
  quantization, with separation of variables, on a K\"ahler
  manifold. The formula is given in terms of differential operators,
  parametrized by acyclic combinatorial graphs.
\end{abstract}

\chapter{Introduction}

Among the first to systematically develop the notion of deformation
quantization were Bayen, Flato, Fronsdal, Lichnerowicz and
Sternheimer. In \cite{MR0496157} and \cite{MR0496158}, they developed
the notion of quantization as a deformation of the commutative algebra
of classical observables through a family of non-commutative products
$\star_h$, parametrized by a real parameter $h$, and gave an
independent formulation of quantum mechanics using this notion.

As opposed to other approaches to quantization, such as geometric
quantization, the theory of deformation quantization does not attempt
to construct a space of quantum states, but focuses the algebraic
structure of the space of observables.

Much work has been done on the theory of deformation quantization, and
it's formal counterpart, where $h$ is interpreted as a formal
parameter. In its most general context, deformation quantization is
studied on Poisson manifolds. In \cite{MR2062626}, Kontsevich proves
the existence of a formal deformation quantization on any Poisson
manifold. Moreover, he gives a formula for a deformation quantization
of any Poisson structure on $\setR^n$. His formula describes the star
product in terms of bidifferential operators parametrized by graphs
and with coefficients given by integrals over appropriate
configuration spaces.  This bears resemblance in flavour to the
construction presented in this paper, which is also based on a certian
interpretation of graphs as differential operators.

Other significant constructions of star products include the
geometrical construction by Fedosov in \cite{MR1293654}, where he
constructs a deformation quantization on an arbitrary symplectic
manifold. Moreover, we should mention the work of Schlichenmaier
\cite{MR1805922}, where he uses the theory of Toeplitz operators to
construct a deformation quantization on any compact K\"ahler manifold.

The question of existence and classification of deformation
quantizations on an arbitrary symplectic manifold was solved by De
Wilde and Lecomte in \cite{MR728644}, where they show that equivalence
classes of star products are classified by formal cohomology
classes. On K\"ahler manifolds, existence and classification was
addressed by Karabegov in \cite{MR1408526}, where he proves that
deformation quantizations with separation of variables are classified,
completely and not only up to equivalence, by closed formal
$(1,1)$-forms, which he calls formal deformations of the K\"ahler
form.  In this paper, we shall be dealing exclusively with deformation
quantizations, with separation of variables, on K\"ahler manifolds.

In this setting, Berezin \cite{MR0395610} originally wrote down
integral formulas for a star product, but he had to make severe
assumptions on the K\"ahler manifold. By interpreting Berezin's
integral formulas formally, and studying their asymptotic behavior,
Reshetikhin and Takhtajan \cite{MR1772294} gave an explicit formula,
in terms of Feynman graphs,
for a formal deformation quantization on any K\"ahler manifold.

Reshetikhin and Takhtajan applied the method of stationary phase to
Berezin's integrals to obtain the asymptotic expansion, and the
description in terms of Feynman graphs arises in a natural way through
this approach. However, the graphs produced by the expansion of
Berezin's integrals have relations among them, expressing fundamental
identities on the K\"ahler manifold. Moreover, the expansion produces
disconnected graphs which prevent the star product from being
normalized.

Using the general existence of a unit,
Reshetikhin and Takhtajan defined a normalized version of the star
product. The coefficients of the unit for the non-normalized star
product can be determined inductively by solving the defining
equations for the unit, but this approach does not yield an explicit
formula for the unit in terms of Feynman graphs, and consequently such
a formula for the normalized star product was not given.

The present paper grew out of an attempt to find an explicit formula
for this normalized star product of Reshetikhin and Takhtajan in terms
of graphs. The crucial observation is that relations among the graphs,
as well as the fact that the star product is not normalized, are caused by
graphs with cycles. 

Given a formal deformation of the K\"ahler form,
we present a local formula for a star product on a K\"ahler manifold
by interpreting graphs as differential operators in a way which is
very similar to \cite{MR1772294}, but we restrict attention to graphs
without cycles. We show that the formula in fact defines a global
deformation quantization on the K\"ahler manifold, with classifying
Karabegov form given by the formal deformation of the K\"ahler form
used in the definition of the star product. Thus our construction
gives a local formula for any deformation quantization, with
separation of variables, on a K\"ahler manifold.

The main result of the paper is stated in the following theorem.

\begin{theorem}
  \label{thm:4}
  The unique formal deformation quantization on $M$ with Karabegov
  form $\omega$ is given by the local formula
  \begin{align*}
    f_1 \star f_2 = \sum_{G \in \mathcal{A}_2} \frac{1}{\abs{\Aut(G)}}
    \Gamma_G(f_1, f_2) h^{W(G)},
  \end{align*}
  for any functions $f_1$ and $f_2$ on $M$.
\end{theorem}

The various ingredients of this theorem and the formula will be
introduced in the following sections, as the definitions of graphs and
their partition functions are a bit more involved than what is
suitable for the introduction. At this point, let us instead give an
overview of the organization of the paper and point to the sections
where the relevant notions are introduced.

In the next section, we introduce the notion of deformation
quantization, and establish some basic notation. Moreover, we recall
how the classifying Karabegov form of a star product with separation
of variables is calculated. Then, we move on to describe the relevant
types of graphs in section \ref{cha:graphs}, where the set
$\mathcal{A}_2$ of acyclic weighted graphs and the total weight
$W(G)$ of a graph $G$ are also defined. The interpretation of a graph
$G \in \mathcal{A}_2$ as a bidifferential operator is defined in
\hbox{section \ref{cha:partition-functions}} through the partition
function $\Gamma_G(f_1, f_2)$, which depends on a choice of local
holomorphic coordinates and a formal deformation $\omega$ of the
K\"ahler form.  It is by no means clear that the formula in
\Fref{thm:4} defines an associative product, and we will need to
rewrite the formula in terms of partition functions of graphs with
more structure to prove associativity. This is done in section
\ref{cha:altern-expr-d}, and associativity is then proved in section
\ref{cha:associativity}, using only combinatorial
considerations. Finally, the Karabegov form of the local product
defined by the formula in \Fref{thm:4} is calculated in section
\ref{cha:coord-invar-class}, and the proof of the theorem is
concluded.

\chapter{Deformation Quantization and K\"ahler Manifolds}

A Poisson structure on a smooth manifold $M$ is a skew-symmetric
bilinear map $\{\cdot, \cdot \} \colon C^\infty(M) \times C^\infty(M)
\to C^\infty(M)$ satisfying the Jacobi identity and the Leibniz rule,
\begin{align*}
  \{f_1, f_2f_3 \} = \{ f_1, f_2 \} f_3 + f_2 \{f_1, f_3\},
\end{align*}
with respect to multiplication of functions.

Deformation quantization makes sense for general Poisson
manifolds. Let $h$ be a formal parameter, and consider the space
$C_h^\infty(M) = C^\infty(M) [[h]]$ of formal power series in $h$ with
coefficients in smooth complex-valued functions on the manifold.
\begin{definition}
  \label{def:1}
  A formal deformation quantization of a Poisson manifold $M$, is an
  associative and $\setC[[h]]$-bilinear product on $C_h^\infty(M)$,
  \begin{align*}
    f_1 * f_2 = \sum_k C_k(f_1, f_2) h^k,
  \end{align*}
  which satisfies
  \begin{align*}
    C_0(f_1, f_2) = f_1 f_2 \qquad \text{and} \qquad C_1(f_1, f_2) -
    C_1(f_2, f_1) = -i \{ f_1, f_2 \},
  \end{align*}
  for any functions $f_1$ and $f_2$ on $M$.
\end{definition}

Very often, extra conditions are imposed on a deformation
quantization.  For instance, the operators $C_k$ are often required to
be bidifferential operators, in which case the star product is said to
be \emph{differential}. Moreover, we say that the star product is
\emph{normalized} if $1 * f = f * 1 = f$, for any function $f$, or
equivalently if $C_k(1, f) = C_k(f,1) = 0$, for $k \geq 1$.

An important source of Poisson manifolds are symplectic manifolds. Any
symplectic manifold $(M, \omega)$, where $\omega \in \Omega^2(M)$ is
non-degenerate and closed, has a canonical Poisson structure defined
by
\begin{align*}
  \{ f_1, f_2 \} = \omega (X_{f_1}, X_{f_2}),
\end{align*}
where $X_f$ denotes the Hamiltonian vector field of a function $f \in
C^\infty(M)$, which is the unique vector field satisfying $df =
\omega(X_f, \cdot)$.

A K\"ahler manifold is a symplectic manifold $(M, \omega)$ equipped
with a compatible complex structure. If $J$ denotes the corresponding
integrable almost complex structure, then compatibility means that
\begin{align*}
  g(X, Y) = \omega(X, JY)
\end{align*}
defines a Riemannian metric on $M$. A deformation quantization $*$ on
a K\"ahler manifold is said to be \emph{with separation of variables}
if $f_1 * f_2 = f_1f_2$, whenever $f_1$ is holomorphic or $f_2$
anti-holomorphic.

We shall be working exclusively with deformation quantizations, with
separation of variables, on K\"ahler manifolds, so for the rest of the
paper, let $M$ be an arbitrary $m$-dimentional K\"ahler manifold with
complex structure $J$, Riemannian metric $g$ and symplectic form
$\omega_{-1}$.

A formal deformation of the K\"ahler form $\omega_{-1}$ is a formal
two-form,
\begin{align*}
  \omega = \frac{1}{h} \omega_{-1} + \omega_0 + \omega_1h + \omega_2
  h^2 + \cdots,
\end{align*}
where each $\omega_k$ is a closed form of type $(1,1)$. Karabegov has
shown that deformation quantizations with separation of variables on
the K\"ahler manifold $M$, are parametrized by such formal
deformations \cite{MR1408526}.

Let us briefly recall how the Karabegov form of a star product $*$ is
calculated. Let $z^1, \ldots, z^m$ be local holomorphic coordinates on
an open subset $U$ of $M$, and suppose that $\Psi^1, \ldots, \Psi^m$
is a set of formal functions on $U$,
\begin{align*}
  \Psi^k = \frac{1}{h} \Psi^k_{-1} + \Psi^k_0 + \Psi^k_1 h + \Psi^k_2
  h^2 + \cdots,
\end{align*}
satisfying
\begin{align*}
  \Psi^k * z^l - z^l * \Psi^k = \delta^{kl}.
\end{align*}
Then the classifying Karabegov form of $*$, which is a global form on
$M$, is given by $\omega \vert_U = - i \bar \d (\sum_k \Psi^k dz^k)$
on the coordinate neighborhood $U$.

For the rest of the paper, $\omega$ will denote a fixed formal
deformation of the K\"ahler form. Also, since we shall be working a
lot in local coordinates, we fix a set of holomorphic coordinates
$z^1, \ldots z^m$ on an open and contractible subset $U$ of $M$.

Choose a formal potential of the form $\omega$ on $U$, that is, choose
a formal function
\begin{align*}
  \Phi = \frac{1}{h} \Phi_{-1} + \Phi_0 + \Phi_1h + \Phi_2 h^2 +
  \cdots,
\end{align*}
such that $\omega \vert_U = i \d \bar \d \Phi$. The existence of a
potential is guaranteed by the fact that $\omega$ is closed and of
type $(1,1)$.

On $U$, the K\"ahler metric is given by the matrix with entries
\begin{align*}
  g_{p\bar q} = g \Big( \dd{}{z^p}, \dd{}{\bar z^q} \Big) = \frac{\d^2
    \Phi_{-1}}{\d z^p \d \bar z^q}.
\end{align*}
Of course this matrix is invertible, and we denote the entries of the
inverse by $g^{\bar q p}$. With this notation, the Poisson bracket is
given by
\begin{align*}
  \{ f_1, f_2 \} = i \sum_{pq} g^{\bar q p} \Big( \dd{f_1}{\bar z^q}
  \dd{f_2}{z^p} - \dd{f_1}{z^p}\dd{f_2}{\bar z^q} \Big).
\end{align*}
Having established the basic notions, let us define the class of
graphs that we shall be working with.

\chapter{Graphs}
\label{cha:graphs}

A directed graph consists of vertices connected by directed edges. If
$G$ is a graph, the set of vertices is denoted by $V_G$ and the set of
edges by $E_G$. The way edges are connected to vertices is encoded by
two maps $h_G, t_G \colon E_G \to V_G$ specifying the \emph{head} and
\emph{tail} of each edge.

An edge is said to be a loop if it has the same head and tail, and two
edges are said to be \emph{parallel} if they connect the same
vertices. A \emph{cycle} is a path that starts and ends at the same
vertex.

We will allow parallel edges in our graphs, but not cycles. In
particular, we do not allow any loops.

A graph without cycles is said to be \emph{acyclic}, and must have at
least one vertex, called a \emph{source}, with only outgoing edges and
at least one \emph{sink} with only incoming edges. We will consider
graphs with a destinguished set of numbered vertices, which we will
call \emph{external}. The rest of the vertices are called
\emph{internal}. The sets of external and internal vertices are
denoted $\Ext(G)$ and $\Int(G)$, respectively. Only an external vertex
is allowed to be a source or a sink, and we require that the first
external vertex is a source and that the last is a sink.

All graphs will be \emph{weighted}, in the sense that each internal
vertex is assigned a weight from the subset $\{-1, 0, 1, 2, \ldots \}$
of integers, and we shall require that vertices of weight -1 have
degree at least three.

The weight of a vertex $v$ is denoted by $w(v)$.  If $G$ is a graph,
we define the total weight of the graph by
\begin{align*}
  W(G) = \abs{E_G} + \sum_{v \in \Int(G)} w(v).
\end{align*}

\begin{figure}[h]
  \centering
  \includegraphics{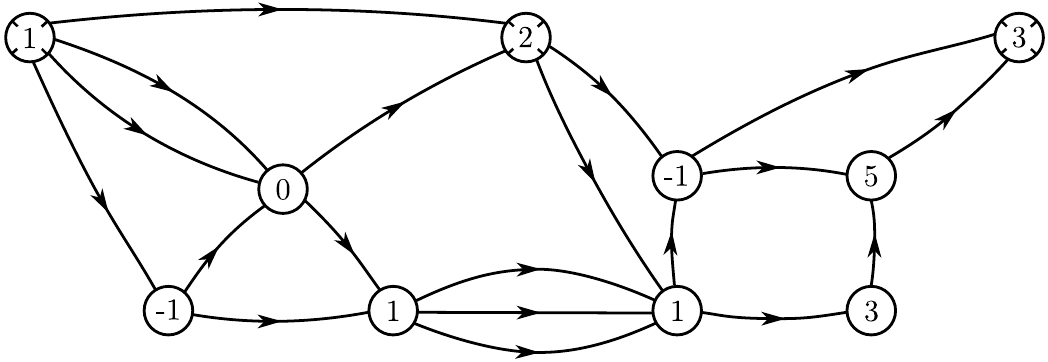}
  \caption{A weighted acyclic graph of total weight 29.}
  \label{fig:1}
\end{figure}

An isomorphism of two graphs is a bijective mapping of vertices to
vertices and edges to edges, preserving the way vertices are connected
by edges, and preserving the external edges and their
numbering. Moreover, an isomorphism should preserve the weights on
internal vertices. If $G$ is a graph, then the set of automorphisms is
denoted by $\Aut(G)$.

The set of isomorphism classes of finite, acyclic and weighted graphs
with $n$ external vertices is denoted by $\mathcal{A}_n$. The subset
of graphs with total weight $k$ is denoted by $\mathcal{A}_n(k)$.

\chapter{Partition Functions}
\label{cha:partition-functions}

In this section, we define the partition function $\Gamma_G(f_1,
\ldots, f_n) \in C^\infty(U)$, for any graph $G \in \mathcal{A}_n$ and
any functions $f_1, \ldots, f_n$ on $U$.

Let us first introduce some notation. If $f \in C^\infty(U)$ is a
function, we define, for each pair of non-negative integers $p$ and
$q$, a covariant tensor $f^{(p,q)}$ on $U$ of type $(p,q)$ by
\begin{align*}
  f^{(p,q)} \Big (\dd{}{z^{i_1}}, \ldots, \dd{}{z^{i_p}}, \dd{}{\bar
    z^{j_1}}, \ldots, \dd{}{\bar z^{j_q}} \Big ) = \frac{\d^{p+q}
    f}{\d z^{i_1} \cdots \d z^{i_p} \d \bar z^{j_1} \cdots \d \bar
    z^{j_q}}.
\end{align*}

Assign to each vertex $v \in V_G$, with $p$ incoming and $q$ outgoing
edges, a tensor by the following rule. If $v$ is the $k$-th external
vertex, we associate the tensor $f_k^{(p,q)}$, and if $v$ is an
internal vertex of weight $w$, we associate the tensor
$-\Phi_w^{(p,q)}$.

Then, we define the partition function $\Gamma_G(f_1, \ldots, f_n)$ to
be the function given by contracting the tensors associated to each
vertex, using the K\"ahler metric, as prescribed by the edges of the
graph. Since the tensors are completely symmetric, this contraction is
well-defined.

Notice that the partition function depends on the deformation $\omega$
of the K\"ahler form, but not on choice of potential $\Phi$. This is
because every internal vertex has at least one incoming and outgoing
edge, and so the potential is differentiated at least once in both a
holomorphic and an anti-holomorphic direction.

Using the partition functions of graphs, we define the following
formal multi-differential operator
\begin{align*}
  D(f_1, \ldots, f_n) = \sum_{G \in \mathcal{A}_n}
  \frac{1}{\abs{\Aut(G)}} \Gamma_G(f_1, \ldots, f_n) h^{W(G)}.
\end{align*}
If we define the multi-differential operators
\begin{align*}
  D_k(f_1, \ldots, f_n) = \sum_{G \in \mathcal{A}_n(k)}
  \frac{1}{\abs{\Aut(G)}} \Gamma_G(f_1, \ldots, f_n),
\end{align*}
then $D$ is given by the formal power series of operators $D = \sum_k
D_k h^k$.

The first result of the paper is stated in the following theorem.
\begin{theorem}
  \label{thm:1}
  The product
  \begin{align*}
    f_1 \star f_2 = D(f_1, f_2) = \sum_k D_k(f_1, f_2) h^k
  \end{align*}
  defines a normalized formal deformation quantization with separation
  of variables on the coordinate neighborhood $U$.
\end{theorem}
 
Since the only graph with two external vertices and total weight zero
is the graph with no edges and no internal vertices, we clearly have
\begin{align*}
  D_0(f_1, f_2) = f_1 f_2.
\end{align*}
Moreover, there is only one graph of total weight one, namely the
graph with no internal vertices and only one edge connecting the two
external vertices. Therefore

\begin{align*}
  D_1(f_1, f_2) = \sum_{pq} g^{\bar q p} \frac{\d f_1}{\d \bar z^q}
  \frac{\d f_2}{\d z^p},
\end{align*}
and we get that
\begin{align*}
  D_1(f_1, f_2) - D_1(f_2, f_1) = -i \{f_1, f_2\},
\end{align*}
as required of a deformation quantization.

Note, that the expression for the star product is with separation of
variables, since the first external vertex has no incoming edges, and
the second has no outgoing. Also, note that the star product is
normalized, since any graph of total weight higher than zero must have
edges, and therefore the external vertices must have degree at least
one.

The only part of \Fref{thm:1} that remains to be proved is
associativity of the star product. We will prove this by combinatorial
arguments involving certain modifications on graphs.

Since the size of the automorphism group of a graph does not behave
well under these modifications, the expression for the star product
given above is not suitable to work with. Therefore, we need to find a
different expression which behaves better when modifying the graphs.

\chapter{Alternative Expression for the Operator $D$}
\label{cha:altern-expr-d}

Let us be a little more explicit in writing out the partition function
$\Gamma$. To this end, we need to introduce further structure on
graphs.

A \emph{labelling} $l$ of a graph $G \in \mathcal{A}_n$ is an
assignment of indices to the incoming and outgoing edges of each
vertex of the graph. If $v$ is a vertex and $e$ is an incident edge,
then the index specified by the labelling is an integer in the set
$\{1, \ldots, m\}$ and is denoted by $l(v,e)$.

\begin{figure}[h]
  \centering
  \includegraphics{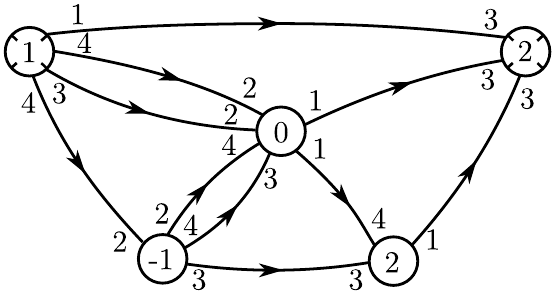}
  \caption{A labelled graph.}
  \label{fig:2}
\end{figure}

An isomorphism of labelled graphs is of course an isomorphism
preserving the labels. The set of labellings of a graph $G$ is denoted
by $\mathcal{L}(G)$, and the set of isomorphism classes of labelled
graphs with $n$ external edges is denoted by $\mathcal{L}_n$.

Let us introduce a partition function $\Lambda^l_G(f_1, \ldots, f_n)$
of a labelled graph $G$ with labelling $l$. For notational
convenience, we first define a function $F_{f_1, \ldots, f_n} \colon
V_G \sqcup E_G \to C^\infty(U)$, which assigns a function to each
vertex and edge of the graph.

Let $v$ be a vertex of $G$, with $p$ incoming and $q$ outgoing edges,
and suppose that the incoming edges are labelled with indices $i_1,
\ldots, i_p$, and the outgoing vertices are labelled with indices
$j_1, \ldots, j_q$. If $v$ is the $k$'th external vertex, then we
define
\begin{align*}
  F_{f_1, \ldots, f_n}(v) = \frac{\d^{p+q} f_k}{\d z^{i_1} \cdots \d
    z^{i_p} \d \bar z^{j_1} \cdots \d \bar z^{j_q}}.
\end{align*}
If $v$ is an internal vertex with weight $w$, then we define
\begin{align*}
  F_{f_1, \ldots, f_n}(v) = - \frac{\d^{p+q} \Phi_w}{\d z^{i_1} \cdots
    \d z^{i_p} \d \bar z^{j_1} \cdots \d \bar z^{j_q}}.
\end{align*}
Notice that this does not depend on the choice of potential, since
internal vertices have at least one incoming and outgoing edge.
Finally, if $e$ is an edge from $u$ to $v$, and we let $s = l(u, e)$
and $r = l(v, e)$, then we define $F_{f_1, \ldots, f_n}(e) = g^{\bar
  sr}$.

Using this, we define
\begin{align*}
  \Lambda^l_G(f_1, \ldots, f_n) = \Big (\prod_{v \in V_G} F_{f_1,
    \ldots, f_n}(v) \Big ) \Big (\prod_{e \in E_G} F_{f_1, \ldots,
    f_n} (e) \Big ).
\end{align*}
From the definition of $\Gamma_G$, it should be obvious that
\begin{align*}
  \Gamma_G(f_1, \ldots, f_n) = \sum_{l \in \mathcal{L}(G)}
  \Lambda^l_G(f_1, \ldots, f_n).
\end{align*}
Therefore, we have the following expression for $D$,
\begin{align*}
  D(f_1, \ldots, f_n) = \sum_{G \in \mathcal{A}_n} \sum_{l \in
    \mathcal{L}(G)} \frac{1}{\abs{\Aut(G)}} \Lambda^l_G(f_1, \ldots,
  f_n) h^{W(G)}.
\end{align*}

The fact that we have written out $D$ in terms on $\Lambda$ will aid
us in later arguments. However, the size of the automorphism group
does not behave well when modifying graphs as we shall later
do. Therefore, we will need to rewrite our expression for $D$ further.

If $G$ is a graph in $\mathcal{A}_n$, a \emph{circuit structure} on
$G$ is a total order, for each vertex of $G$, of the incoming as well
of the outgoing edges of that vertex. This gives rise to a numbering
of the incoming as well as the outgoing edges at each vertex, and if
$v$ is a vertex of $G$ with an incident edge $e$, the circuit
structure therefore specifies a natural number $c(v,e)$.  An
isomorphism of circuit graphs is an isomorphism which preserves the
ordering on the incoming and outgoing edges at each vertex.

\begin{figure}[h]
  \centering
  \includegraphics{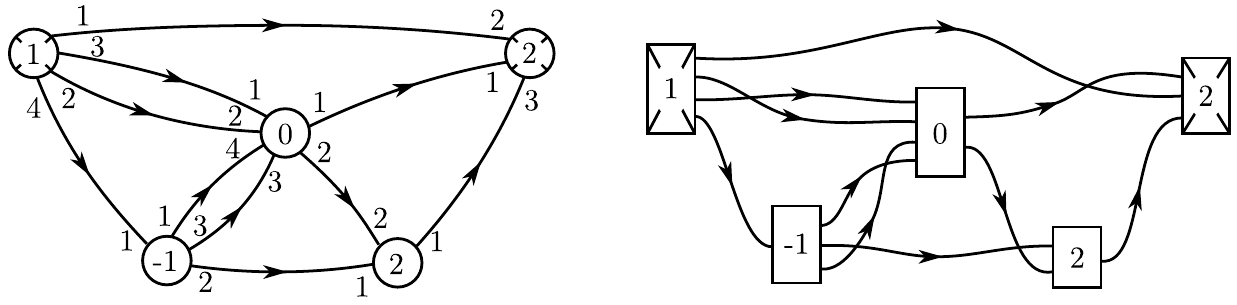}
  \caption{Different representations of a circuit graph.}
  \label{fig:3}
\end{figure}

\Fref{fig:3} shows two ways of representing a circuit structure
graphically. The latter, with rectangular vertices, is usually
preferred. This also motivates the name circuit structure, as it
resembles a diagram of an electrical circuit, where a number of chips,
with input and output pins, are connected by wires. This analogy is
also supported by the fact that our graphs are acyclic.

The set of circuit structures on $G$ is denoted by $\mathcal{C}(G)$,
and the set of isomorphism classes of circuit graphs with $n$ external
vertices is denoted by $\mathcal{C}_n$.

Very often, we shall be working with graphs equipped with both a
labelling and a circuit structure, and we will need to enforce a
certain compatibility between the two structures.

If $G \in \mathcal{A}_n$ is a graph equipped with a labelling $l$ and
a circuit structure $c$, we say that $l$ and $c$ are \emph{compatible}
if for any vertex $v$ and any two edges $e$ and $e'$ incident to $v$,
with the same orientation, we have that $c(v,e) \leq c(v, e')$ implies
$l(v, e) \leq l(v, e')$. In other words, the incoming edges of a
vertex should be labelled ascendingly with respect to the ordering
given by the circuit structure, and likewise for the outgoing edges.

\begin{figure}[h]
  \centering
  \includegraphics{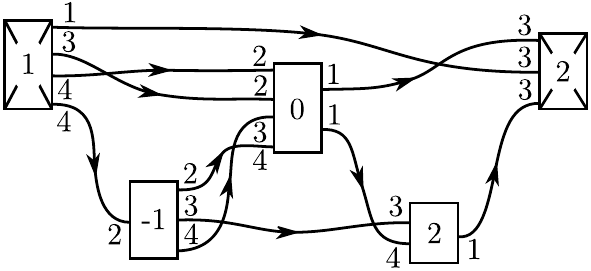}
  \caption{A labelled circuit graph.}
  \label{fig:4}
\end{figure}

If $G$ is a graph with labelling $l$, the set of compatible circuit
structures is denoted by $\mathcal{C}(G, l)$. The set of isomorphism
classes of labelled graphs with a compatible circuit structure is
denoted $\mathcal{L}^C_n$.

Given a labelled graph, the number of compatible circuit structures
will be important to us. To calculate this, we will need some
notation.

Recall that a multi-index is an $m$-tuple $\alpha = (\alpha_1, \ldots,
\alpha_m) \in \setN^m_0$. The length of $\alpha$ is defined to be
$\abs{\alpha} = \alpha_1 + \ldots + \alpha_m$, and we define $\alpha!
= \alpha_1! \cdots \alpha_m!$. A labelling of a graph assigns two
multi-indices to each vertex in a canonical way. More precisely, if
$G$ is a graph with labelling $l$, then we have two canonically
defined maps $\alpha_l, \beta_l \colon V_G \to \setN_0^m$. If $v$ is a
vertex of $G$, then the multi-index $\alpha_l(v)$ counts the number of
occurrences of each label among the incoming edges of $v$. Similarly,
the multi-index $\beta_l(v)$ counts the occurrences of each label
among the outgoing edges.

Now, given the graph $G$ with labelling $l$, the number of compatible
circuit structures is given by
\begin{align*}
  C(G, l) = \prod_{v \in V_G} \alpha_l(v)! \beta_l(v)!.
\end{align*}
Using this, we can rewrite the formula for the operator $D$ as
\begin{align*}
  D(f_1, \ldots, f_n) = \sum_{G \in \mathcal{A}_n} \sum_{l \in
    \mathcal{L}(G)} \sum_{c \in \mathcal{C}(G,l)}
  \frac{1}{\abs{\Aut(G)}C(G, l)} \Lambda_G^l(f_1, \ldots,
  f_n)h^{W(G)},
\end{align*}
since the circuit structure does not influence on the value of the
partition function.

Suppose that $G \in \mathcal{A}_n$ is any graph with $n$ external
edges. If we pick a labelling $l$ and a compatible circuit structure
$c$, then $(G, l, c)$ represents an element of $\mathcal{L}^C_n$. If
we choose a different labelling $l'$ and circuit structure $c'$ on
$G$, then $(G, l', d')$ represents the same isomorphism class in
$\mathcal{L}^C_n$ if and only if there exists an automorphism of $G$,
which sends the labelling $l$ to $l'$ and the circuit structure $c$ to
$c'$. Thus, we have proved the following proposition

\begin{proposition}
  \label{prop:1}
  The operator $D$ is given by
  \begin{align*}
    D(f_1, \ldots, f_n) = \sum_{G \in \mathcal{L}^C_n} \frac{1}{C(G)}
    \Lambda_G(f_1, \ldots, f_n) h^{W(G)},
  \end{align*}
  for any functions $f_1, \ldots, f_n$.
\end{proposition}

As we shall often do when the additional structure is clear from the
context, we have omitted the labelling from the notation in this
proposition.

\chapter{Associativity of the Star Product}
\label{cha:associativity}

With the alternative expression for the operator $D$, given in
\Fref{prop:1}, we are ready to prove associativity of the star
product. This is an immediate corollary of the following theorem.

\begin{theorem}
  \label{thm:2}
  We have
  \begin{align*}
    D(f_1, D(f_2, f_3)) = D(f_1, f_2, f_3) = D(D(f_1, f_2), f_3),
  \end{align*}
  for any functions $f_1$, $f_2$ and $f_3$.
\end{theorem}

We shall only prove the first equality of this theorem. The second
equality follows by analagous arguments.

To prove \Fref{thm:2}, we must have a better understanding of the
expression $D(f_1, D(f_2, f_3))$. Writing out this expression, we have
\begin{align*}
  D(f_1, D(f_2, f_3)) &= \sum_{G_1 \in \mathcal{L}^C_2} \sum_{G_2 \in
    \mathcal{L}^C_2} \frac{1}{C(G_1) C(G_2)} \Lambda_{G_1}(f_1,
  \Lambda_{G_2}(f_2, f_3)) h^{W(G_1)} h^{W(G_2)},
\end{align*}
and we see that $\Lambda_{G_1}(f_1, \Lambda_{G_2}(f_2, f_3))$ is the
crucial part to understand.

Before we prove \Fref{thm:2}, let us illustrate, with an example, how
graphs in the expression for $D(f_1, f_2, f_3)$ arise from $D(f_1,
D(f_2, f_3))$.

\begin{example}
  \label{ex:1}
  Suppose that we have two graphs $G_1$ and $G_2$ in
  $\mathcal{L}^C_2$, as depicted in \Fref{fig:5}.

  \begin{figure}[h]
    \centering
    \includegraphics{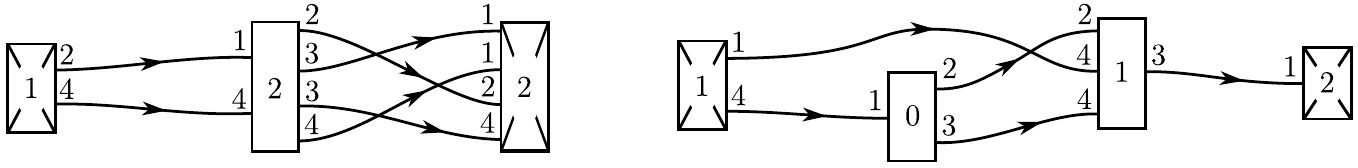}
    \caption{The graphs $G_1$ and $G_2$.}
    \label{fig:5}
  \end{figure}
  
  We think of $G_2$ as representing a term of the inner $D$ in $D(f_1,
  D(f_2, f_3))$, and $G_1$ as representing a term of the outer
  $D$. More precisely, we let
  \begin{align*}
    \hat f = \Lambda_{G_2}(f_2, f_3) = \frac{\d^2 f_1}{\d \bar z^1 \d
      \bar z^4} \frac{\d^3 \Phi_0}{\d z^1 \d \bar z^2 \d \bar z^3}
    \frac{\d^4 \Phi_1}{\d z^2 \d z^4 \d z^4 \d \bar z^3} \frac{\d
      f_2}{\d z^1} g^{\bar 1 4} g^{\bar 4 1} g^{\bar 2 2} g^{\bar 34}
    g^{\bar 3 1},
  \end{align*}
  and we want to calculate the partition function
  \begin{align*}
    \Lambda_{G_1}(f_1, \hat f) = - \frac{\d^2 f_1}{\d \bar z^2 \d \bar
      z^4} \frac{\d^6 \Phi_{2}}{\d z^1 \d z^4 \d \bar z^2 \d \bar z^3
      \d \bar z^3 \d \bar z^4} \frac{\d^4 \hat f}{\d z^1 \d z^1 \d z^2
      \d z^4} g^{\bar 21} g^{\bar 44} g^{\bar 22} g^{\bar 3 1} g^{\bar
      34} g^{\bar 4 1}.
  \end{align*}
  Informally, we have the picture in \Fref{fig:6} in mind as a
  graphical representation of this expression.
  \begin{figure}[h]
    \centering
    \includegraphics{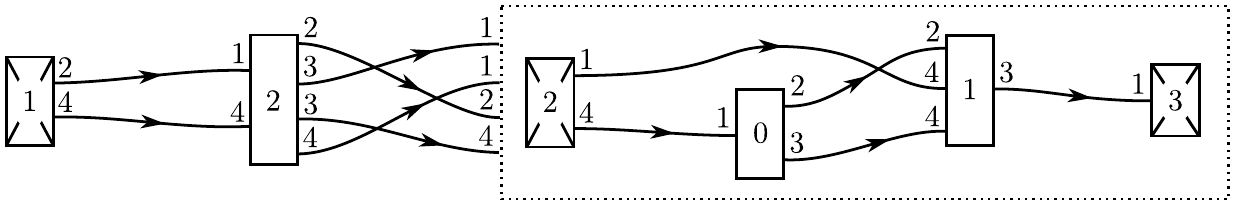}
    \caption{Calculating $\Lambda_{G_1}(f_1, \Lambda_{G_2}(f_2,
      f_3))$.}
    \label{fig:6}
  \end{figure}
  
  Since $\hat f$ is given by a product, the Leibniz rule says that
  $\frac{\d^4 \hat f}{\d z^1 \d z^1 \d z^2 \d z^4}$ is given by a sum,
  where each term represents a certain way of distributing the partial
  derivatives among the factors.
  
  Let us focus on one such term, say the one where the first and the
  third partial derivative from the left hit the factor $\frac{\d^3
    \Phi_0}{\d z^1 \d \bar z^2 \d \bar z^3}$, the second derivative
  hits the factor $\frac{\d^2 f_1}{\d \bar z^1 \d \bar z^4}$, and the
  fourth hits the factor $g^{\bar 34}$. That term is then given by
  \begin{align*}
    \frac{\d^3 f_1}{\d z^1\d \bar z^1 \d \bar z^4} \frac{\d^5
      \Phi_0}{\d z^1 \d z^1 \d z^2\d \bar z^2 \d \bar z^3} \frac{\d^4
      \Phi_1}{\d z^2 \d z^4 \d z^4 \d \bar z^3} \frac{\d f_2}{\d z^1}
    g^{\bar 1 4} g^{\bar 4 1} g^{\bar 2 2} \frac{\d g^{\bar 34}}{\d
      z^4} g^{\bar 3 1}.
  \end{align*}
  But partial derivatives of the inverse metric can be easily
  expressed in terms of partial derivatives of the metric, as in
  \begin{align*}
    \frac{\d g^{\bar 3 4}}{\d z^4} = - \sum_{pq} g^{\bar 3 p} \frac{\d
      g_{p \bar q}}{\d z^4} g^{\bar q 4} = - \sum_{pq} g^{\bar 3 p}
    \frac{\d^3 \Phi_{-1}}{\d z^4 \d z^p \d \bar z^q} g^{\bar q 4} .
  \end{align*}
  If we choose particular values, say $p = 1$ and $q = 2$, for the
  summation variables, then we arrive at
  \begin{align*}
    - \frac{\d^3 f_1}{\d z^1\d \bar z^1 \d \bar z^4} \frac{\d^5
      \Phi_0}{\d z^1 \d z^1 \d z^2\d \bar z^2 \d \bar z^3} \frac{\d^4
      \Phi_1}{\d z^2 \d z^4 \d z^4 \d \bar z^3} \frac{\d^3
      \Phi_{-1}}{\d z^1 \d z^4 \d \bar z^2} \frac{\d f_2}{\d z^1}
    g^{\bar 1 4} g^{\bar 4 1} g^{\bar 2 2} g^{\bar 3 1} g^{\bar 2 4}
    g^{\bar 3 1}
  \end{align*}
  as an example of what terms in the expression for $\frac{\d^4 \hat
    f}{\d z^1 \d z^1 \d z^2 \d z^4}$ look like.
  
  If we insert this into the expression for $\Lambda_{G_1}(f_1, \hat
  f)$ above, we get an example of what terms in the expression for
  $D(f_1, D(f_2, f_3))$ look like. But this particular example can be
  represented graphically by $\Lambda_G{(f_1, f_2, f_3)}$, where $G
  \in \mathcal{L}^C_3$ is the graph shown in \Fref{fig:7}.
  \begin{figure}[h]
    \centering
    \includegraphics{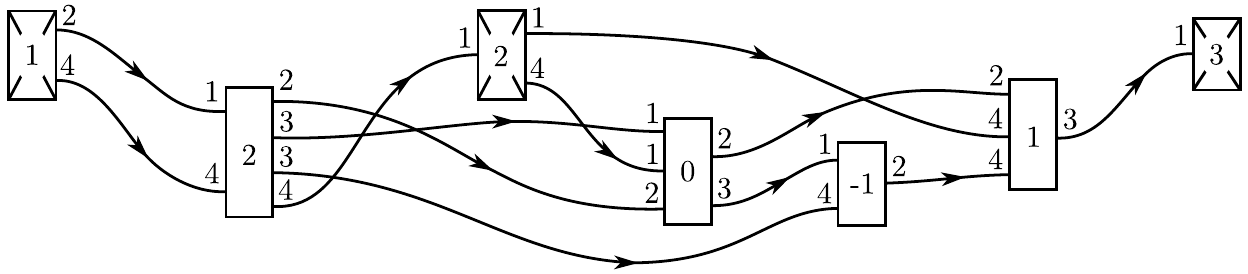}
    \caption{A fusion G of the two graphs $G_1$ and $G_2$.}
    \label{fig:7}
  \end{figure}

\end{example}

With this concrete example in mind, let us turn to more general
considerations. The graph in \Fref{fig:7} is an example of a
\emph{fusion} of the graphs $G_1$ and $G_2$. Let us define this notion
more carefully.

Let $G_1$ and $G_2$ be two graphs in $\mathcal{L}^C_2$. A fusion of
$G_1$ onto $G_2$ is a graph $G \in \mathcal{L}^C_3$ with three
external vertices, obtained through the following procedure.  Cut out
the second external vertex of $G_1$, leaving a collection of labelled
edges with loose ends. Connect each of these loose ends, one at a
time, to the graph $G_2$ in one of two possible ways. The first is to
connect a loose end to one of the vertices of $G_2$, and extend the
circuit structure at the vertex, in any way compatible with the
labelling, to include the newly attached edge. The second possibility
os to attach a loose end to one of the edges of $G_2$. This is done by
adding a vertex of weight -1 on the edge, choosing any labelling of
the two edges incident to the new vertex, attaching the loose end to
the new vertex and choosing a circuit structure at the
vertex. Finally, the first and second external vertices of $G_2$ will
be the second and third external vertex of the fusion, respectively.

Clearly, a fusion of two graphs results in a labelled circuit graph
with 3 external vertices. The set of isomorphism classes of such
graphs, that can be obtained from two graphs $G_1$ and $G_2$ through a
fusion procedure, is denoted $\mathcal{F}(G_1, G_2)$.

Given a labelled circuit graph $G \in \mathcal{L}^C_3$, with three
external vertices, it is natural to ask if this can be obtained as a
fusion of two graphs $G_1$ and $G_2$ in $\mathcal{L}^C_2$.  Moreover,
it is natural to ask how much information about the graphs $G_1$ and
$G_2$ is encoded in a fusion.

Given two vertices $u$ and $v$ of a graph, we say that $v$ is a
\emph{successor} of $u$ if there exists a directed path from $u$ to
$v$. A crucial observation is that when $G_1$ is fused to $G_2$, any
vertex in $G_2$ which is a successor of the first external vertex in
$G_2$ will be a successor of the second external vertex in the
fusion. Moreover, vertices that arose by attaching a loose end to an
edge of $G_2$ will also be successors of the second external
vertex. On the other hand, none of the vertices of $G_1$ will succeed
the second external vertex in the fusion.

These observations can be used to reconstruct nearly all the
information about the structure of the graphs $G_1$ and $G_2$ from a
fusion of these. Moreover, as we shall see, any labelled circuit graph
with three external vertices arises as a fusion.

Suppose that $G \in \mathcal{L}^C_3$ is any labelled circuit graph. We
seek two labelled circuit graphs $G_1$ and $G_2$ such that $G \in
\mathcal{F}(G_1, G_2)$. We can completely determine the isomorphism
class of $G_2$ in $\mathcal{L}^C_2$ by the following procedure. Delete
all vertices from $G$ which are not successors of the second external
vertex, as well as all edges incident to at least one such vertex. The
result may contain vertices of weight -1 and degree 2. These are the
remnants of vertices arising during the fusion when a loose edge end
is connected to an edge of $G_2$. Every such vertex is deleted and the
resulting two loose ends are spliced, forgetting their
labelling. Finally, the second and third external vertices are the
only external vertices left, and they will be the first and second
external vertices in $G_2$, respectively.

In a similar way, we can almost determine the isomorphism class of the
labelled circuit graph $G_1$ by deleting all successors of the second
external vertex in $G$, and all edges between two such successors, and
then connect all the remaining loose edge ends to a new vertex, which
will be the second external vertex of $G_1$. There is however no
canonical way of telling what the circuit structure at the second
external vertex should be.

To deal with this ambiguity, we define an equivalence relation on the
set $\mathcal{L}^C_2$ of labelled circuit graphs with two external
vertices. Consider two graphs $G$ and $G'$ in $\mathcal{L}^C_2$, with
labellings $l$ and $l'$ and circuit structures $c$ and $c'$. We say
that these graphs are equivalent, and we write $G \sim G'$, if there
exists an isomorphism between $G$ and $G'$ which preserves the
labelling at all vertices, and which preserves the circuit structure,
except possibly at the second external vertex. In the discussion
above, the equivalence class of the graph $G_1$ is then completely
determined.

We summarize our findings in the following proposition
\begin{proposition}
  \label{prop:2}
  For any labelled circuit graph $G \in \mathcal{L}^C_3$, there exist
  two labelled circuit graphs $G_1, G_2 \in \mathcal{L}^C_2$ such that
  $G \in \mathcal{F}(G_1, G_2)$. Moreover, the equivalence class of
  $G_1$ is uniquely determined by $G$, and so is the isomorphism class
  of $G_2$.
\end{proposition}

When calculating $D(f_1, D(f_2, f_3)$, we are basically faced with the
task of calculating $\Lambda_{G_1}(f_1, \Lambda_{G_2}(f_2, f_3))$ for
any two labelled circuit graphs $G_1$ and $G_2$. As illustrated in
\Fref{ex:1}, this is given by a sum, where each term can be
represented by a fusion of $G_1$ and $G_2$.

Now suppose that an edge, incident to the second external vertex in
$G_1$ and with label $j$, is attached to a vertex $v$ in $G_2$, and
that $v$ already has $k$ incoming edges with label $j$. Then, when
extending the circuit structure at $v$ to include the newly attached
edge, there are $k+1$ ways of placing the new edge in the ordering of
the incoming edges.

Moreover, suppose that $l$ is the labelling of $G_1$, and let $u$ be
the second external vertex. Then, the size of the equivalence class
$[G_1]$ is given by $\alpha_{l}(u)!$.

These observations suffice to realize that
\begin{align*}
  \sum_{G \in [G_1]} \frac{1}{C(G) C(G_2)} \Lambda_{G}(f_1,
  \Lambda_{G_2}(f_2, f_3)) = \sum_{G \in \mathcal{F}(G_1, G_2)}
  \frac{1}{C(G)}\Lambda_G(f_1, f_2, f_3).
\end{align*}
Since $W(G_1)+W(G_2) = W(G)$ if $G \in \mathcal{F}(G_1, G_2)$, we can
multiply the left-hand side by $h^{W(G_1)}h^{W(G_2)}$ and the
right-hand side by $h^{W(G)}$, and sum over all graphs $G_2 \in
\mathcal{L}^C_2$ and all equivalence classes $[G_1]$ in
$\mathcal{L}^C_2 / \smash{\sim}$ to get

\begin{align*}
  D(f_1, D(f_2, f_3)) &= \sum_{G_1 \in \mathcal{L}^C_2} \sum_{G_2 \in
    \mathcal{L}^C_2} \frac{1}{C(G_1) C(G_2)} \Lambda_{G_1}(f_1,
  \Lambda_{G_2}(f_2, f_3)) h^{W(G_1)} h^{W(G_2)} \\ &= \sum_{[G_1] \in
    \mathcal{L}^C_2 / \smash{\sim}} \sum_{G_2 \in \mathcal{L}^C_2}
  \sum_{G \in \mathcal{F}(G_1, G_2)} \frac{1}{C(G)}\Lambda_G(f_1, f_2,
  f_3) h^{W(G)}.
\end{align*}

But as $[G_1]$ runs through all equivalence classes of
$\mathcal{L}^C_2 / \smash{\sim}$, and $G_2$ runs through
$\mathcal{L}^C_2$, then \Fref{prop:2} tells us that the sets
$\mathcal{F}(G_1, G_2)$ partition the set $\mathcal{L}^C_3$, that is,
they form a collection of disjoint sets whose union is all of
$\mathcal{L}^C_3$. Thus, we conclude that
\begin{align*}
  D(f_1, D(f_2, f_3)) = \sum_{G \in \mathcal{L}^C_3}
  \frac{1}{C(G)}\Lambda_G(f_1, f_2, f_3) = D(f_1, f_2, f_3).
\end{align*}

This proves the first equality of \Fref{thm:2}. The other equality is
proved by similar methods, and therefore the theorem is proved. This
also proves \Fref{thm:1}, which is an immediate corollary of
\Fref{thm:2}.

\chapter{Coordinate Invariance and Classification}
\label{cha:coord-invar-class}

In this section, we prove that the local star product of \Fref{thm:1}
is independent of the coordinates used in its definition. This implies
that it defines a global star product on $M$, and as we shall see, the
Karabegov form of this global star product is given by $\omega$.

The claims above will follow easily from the following theorem.
\begin{theorem}
  \label{thm:3}
  The local star product $\star$ on $U$ has Karabegov form $\omega
  \vert_U$.
\end{theorem}

\begin{proof}
  We shall prove that the formal functions $\Psi^r = \d \Phi / \d z^r$
  satisfy the relation
  \begin{align}
    \label{eq:7}
    \Psi^r \star z^s - z^s \star \Psi^r = \delta^{rs}.
  \end{align}
  This will prove the theorem, since $\omega \vert_U = i\d \bar \d
  \Phi = - i \bar \d (\sum_k \Psi^k dz^k)$.
  
  Clearly, we have $D_0(\Psi^r_{-1}, z^s) - D_0(z^s, \Psi^r_{-1}) = 0$
  and
  \begin{align*}
    D_1(\Psi^r_{-1}, z^s) - D_1(z^s, \Psi^r_{-1}) = - i\{\Psi^r_{-1},
    z^s \} = \delta^{rs},
  \end{align*}
  so the identity \eqref{eq:7} is equivalent to the system of
  identities
  \begin{align*}
    \sum_{l=-1}^{k-1} D_{k-l}(\Psi^r_l, z^s) = 0, \qquad k \geq 1.
  \end{align*}
  To prove this, we define a modification on graphs called a
  \emph{budding}. If $l > -1$ and $G \in \mathcal{A}_2(k-l)$ is a
  graph, we define the budded graph $B(G) \in \mathcal{A}_2(k+1)$ by
  the following procedure. Let $u$ denote first external vertex of $G$
  and convert this into an internal vertex of weight $l$. Then add a
  new first external vertex and connect this to $u$ by a single edge.

  \begin{figure}[h]
    \centering
    \includegraphics{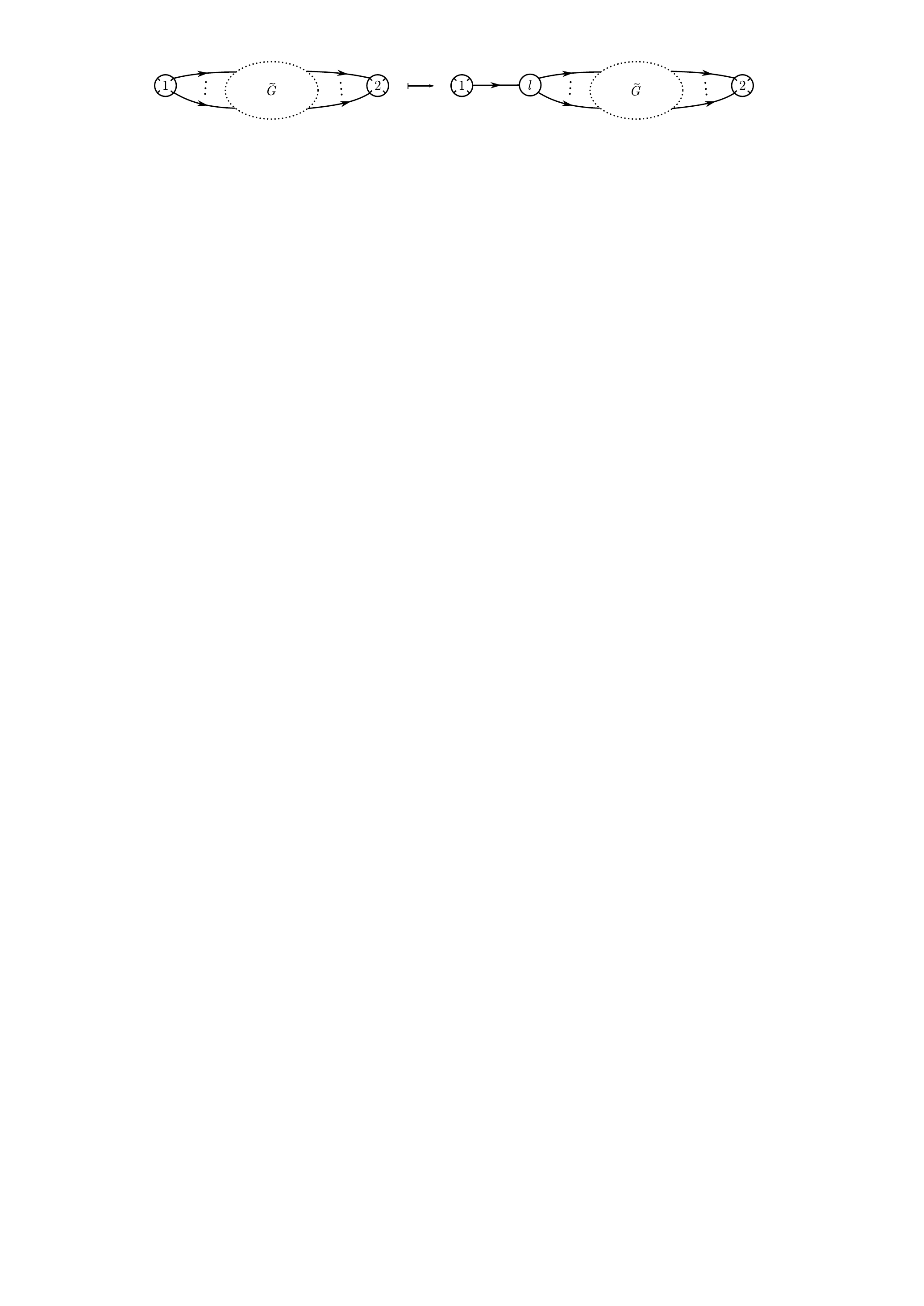}
    \caption{A budding of a graph.}
    \label{fig:8}
  \end{figure}
  
  We had to exclude the case $l=-1$, since the first external vertex
  of $G$ might have degree one, in which case the budded graph would
  not satisfy the rule that internal vertices of weight -1 must have
  degree at least three. However, if we let $\mathcal{A}^1_2(k+1)$ be
  the set of graphs with degree one on the first external vertex, and
  $\mathcal{A}^{>1}_2(k+1)$ be the set of graphs with degree more than
  one on the first external vertex, then the budding construction
  defines a map $B \colon \mathcal{A}^{>1}_2(k+1) \to
  \mathcal{A}^1_2(k+1)$.

  We conclude that the budding construction gives a map
  \begin{align*}
    B \colon \mathcal{A}_2^{>1}(k+1) \cup \bigcup_{l = 0}^{k-1}
    \mathcal{A}_2(k-l) \to \mathcal{A}^1_2(k+1).
  \end{align*}
  Clearly, this map is a bijection, as the inverse map is easily
  constructed. Moreover, it is clear that the budding map preserves
  the size of the automorphism group.

  Now, the crucial property of the budding map is that
  \begin{align*}
    \Gamma_{B(G)}(\Psi^r_{-1}, z^s) = - \Gamma_{G}(\Psi^r_{l}, z^s),
  \end{align*}
  for any graph $G$ in the domain of $B$. Since $B$ is a bijection,
  which preserves the size of the automorphism group, this implies
  that
  \begin{align*}
    &\sum_{l=-1}^{k-1} D_{k-l}(\Psi^r_l, z^s) = \sum_{l=-1}^{k-1}
    \sum_{G \in \mathcal{A}(k-l)} \frac{1}{\abs{\Aut(G)}}
    \Gamma_G(\Psi^r_l, z^s) \\ & \quad = \sum_{G \in
      \mathcal{A}^1_2(k+1)} \frac{\Gamma_G(\Psi^r_{-1},
      z^s)}{\abs{\Aut(G)}} + \sum_{G \in \mathcal{A}^{>1}_2(k+1)}
    \frac{\Gamma_G(\Psi^r_{-1}, z^s)}{\abs{\Aut(G)}} +
    \sum_{l=0}^{k-1} \sum_{G \in \mathcal{A}(k-l)}
    \frac{\Gamma_G(\Psi^r_l, z^s)}{\abs{\Aut(G)}} \\ & \quad =
    \phantom{\sum^1} 0.
  \end{align*}
  This proves the theorem.
\end{proof}

Karabegov's classification has the obvious property that restriction of
a star product to an open subset corresponds to restriction of the
Karabegov form. Therefore, it follows immediately that $\star$ is the
restriction of the unique star product on $M$ with Karabegov form
$\omega$. In particular, the explicit expression given in \Fref{thm:4}
must be independent of the local coordinates used. This finishes the
proof of the main result given in \Fref{thm:4}, which summarizes all
of our findings.

We remark that \Fref{thm:4}
gives an explicit formula, to all orders, of the Berezin star product
with trivial Karabegov form $\frac{1}{h}\omega_{-1}$.

Moreover, in \cite{MR1868597} it was shown that the Berezin-Toeplitz
star product, which is defined on compact K\"ahler manifolds through
asymptotic expansions of products of Toeplitz operators
\cite{MR1805922}, is a differential star product whose opposite star
product is with separation of variables and has
Karabegov form given by $-\frac{1}{h} \omega_{-1} + \rho$, where
$\rho$ denotes the Ricci form on the K\"ahler manifold. Using
\Fref{thm:4}, we can therefore give an explicit formala for the
Berezin-Toeplitz star product to all orders.

The main theorem implies that the operator $D$ is coordinate independent
when applied to two functions, and hence also when applied to three by
\Fref{thm:2}. In fact we conjecture that the general formula for $D$
is coordinate invariant and that there are relations analagous to
\Fref{thm:2}, when applied to a larger collection of functions.

As a closing remark, we think it would be very interesting to use the
formula presented in this paper to try to find invariant expressions
for the star products in terms of covariant derivatives and global
forms.

\subsection*{Acknowledgements}

I would like to thank J\o rgen Ellegaard Andersen and Nicolai
Reshetikhin for bringing the paper \cite{MR1772294} to my attention
and for many enlightening discussions.

\clearpage

\renewcommand{\bibname}{References}


\begin{thebibliography}{10}

\bibitem{MR0496157}
F.~Bayen, M.~Flato, C.~Fronsdal, A.~Lichnerowicz, and D.~Sternheimer.
\newblock Deformation theory and quantization. {I}. {D}eformations of
  symplectic structures.
\newblock {\em Ann. Physics}, 111(1):61--110, 1978.

\bibitem{MR0496158}
F.~Bayen, M.~Flato, C.~Fronsdal, A.~Lichnerowicz, and D.~Sternheimer.
\newblock Deformation theory and quantization. {II}. {P}hysical applications.
\newblock {\em Ann. Physics}, 111(1):111--151, 1978.

\bibitem{MR0395610}
F.~A. Berezin.
\newblock Quantization.
\newblock {\em Izv. Akad. Nauk SSSR Ser. Mat.}, 38:1116--1175, 1974.

\bibitem{MR728644}
Marc De~Wilde and Pierre B.~A. Lecomte.
\newblock Existence of star-products and of formal deformations of the
  {P}oisson {L}ie algebra of arbitrary symplectic manifolds.
\newblock {\em Lett. Math. Phys.}, 7(6):487--496, 1983.

\bibitem{MR1293654}
Boris~V. Fedosov.
\newblock A simple geometrical construction of deformation quantization.
\newblock {\em J. Differential Geom.}, 40(2):213--238, 1994.

\bibitem{MR1408526}
Alexander~V. Karabegov.
\newblock Deformation quantizations with separation of variables on a
  {K}\"ahler manifold.
\newblock {\em Comm. Math. Phys.}, 180(3):745--755, 1996.

\bibitem{MR1868597}
Alexander~V. Karabegov and Martin Schlichenmaier.
\newblock Identification of {B}erezin-{T}oeplitz deformation quantization.
\newblock {\em J. Reine Angew. Math.}, 540:49--76, 2001.

\bibitem{MR2062626}
Maxim Kontsevich.
\newblock Deformation quantization of {P}oisson manifolds.
\newblock {\em Lett. Math. Phys.}, 66(3):157--216, 2003.

\bibitem{MR1772294}
N.~Reshetikhin and L.~A. Takhtajan.
\newblock Deformation quantization of {K}\"ahler manifolds.
\newblock In {\em L. {D}. {F}addeev's {S}eminar on {M}athematical {P}hysics},
  volume 201 of {\em Amer. Math. Soc. Transl. Ser. 2}, pages 257--276. Amer.
  Math. Soc., Providence, RI, 2000.

\bibitem{MR1805922}
Martin Schlichenmaier.
\newblock Deformation quantization of compact {K}\"ahler manifolds by
  {B}erezin-{T}oeplitz quantization.
\newblock In {\em Conf\'erence {M}osh\'e {F}lato 1999, {V}ol. {II} ({D}ijon)},
  volume~22 of {\em Math. Phys. Stud.}, pages 289--306. Kluwer Acad. Publ.,
  Dordrecht, 2000.

\end{thebibliography}


\end{document}